\documentclass[a4paper,twocolumn,11pt,unpublished]{quantumarticle}
\pdfoutput=1

\usepackage[numbers,sort&compress]{natbib}

\usepackage{graphicx} 
\usepackage{amsmath,amssymb,amsthm,mathtools}
\usepackage{qcircuit}
\usepackage{braket}
\usepackage{float}
\usepackage{bm} 
\usepackage[colorlinks=true, allcolors=blue]{hyperref}
\newtheorem{assumption}{Assumption}
\newtheorem{definition}{Definition}
\newtheorem{lemma}{Lemma}
\newtheorem{theorem}{Theorem}

\newcommand{\Tr}{\operatorname{Tr}}
\newcommand{\Var}{\operatorname{Var}}
\newcommand{\E}{\mathbb{E}}
\newcommand{\Id}{\mathsf{Id}}

\newcommand{\poly}{\operatorname{poly}}
\newcommand{\uni}{\mathrm{uni}}
\newcommand{\dis}{\mathrm{dis}}

\newcommand{\gap}{\Delta}
\newcommand{\norm}[1]{\left\lVert #1 \right\rVert}

\begin{document}

\title{A Variational Dissipative Framework for Quantum Algorithms}

\date{May 25, 2026}

\author{Yuan Yao}
\email{yuan.yao@pku.edu.cn}
\affiliation{Center on Frontiers of Computing Studies, Peking University, Beijing 100871, China}
\affiliation{School of Computer Science, Peking University, Beijing 100871, China}

\author{Ruipeng Xing}
\affiliation{Faculty of Information Science and Engineering, Ocean University of China, Qingdao 266100, China}%
\affiliation{Center on Frontiers of Computing Studies, Peking University, Beijing 100871, China}
\affiliation{School of Computer Science, Peking University, Beijing 100871, China}

\author{Yongjian Gu}
\affiliation{Faculty of Information Science and Engineering, Ocean University of China, Qingdao 266100, China}%
\email{yjgu@.ouc.edu.cn}

\author{Yiming Huang}
\email{yiminghwang@gmail.com}
\affiliation{Institute of High Energy Physics, Chinese Academy of Sciences, Beijing 100049, China}
\affiliation{China Center of Advanced Science and Technology, Beijing 100190, China }
\affiliation{Center on Frontiers of Computing Studies, Peking University, Beijing 100871, China}
\affiliation{School of Computer Science, Peking University, Beijing 100871, China}

\author{Xiao Yuan}
\email{xiaoyuan@pku.edu.cn}
\affiliation{Center on Frontiers of Computing Studies, Peking University, Beijing 100871, China}
\affiliation{School of Computer Science, Peking University, Beijing 100871, China}

\begin{abstract}
Dissipation engineering has attracted growing interest as an approach to controlling open quantum systems through engineered system-environment interactions. Standard variational quantum circuits are usually built from unitary operations and therefore explore only a restricted family of states. To go beyond this limitation, we introduce a variational dissipative framework in which ancilla-assisted engineered dissipation is incorporated into parameterized quantum algorithms. In this framework, system-only variational layers are combined with trainable dissipative modules, so that the circuit can prepare a broader class of mixed states through ancilla-assisted nonunitary transformations.
Within this framework, the same ancilla-assisted dissipative block is used in two representative settings with different objectives. For ground-state search, it is integrated into a dissipative variational quantum eigensolver to improve the convergence toward low-energy states. For state recovery, it is trained as a recovery channel to suppress preparation noise and enhance fidelity with the target state. In both cases, the block is realized through parameterized system-ancilla couplings followed by ancilla reset and trace-out.
Our results show that engineered dissipation can be incorporated into variational quantum circuits as a reusable trainable primitive rather than treated only as a source of noise. In this sense, the proposed framework identifies ancilla-assisted dissipative channels as a concrete variational resource that can support both optimization and recovery tasks within a unified design.
\end{abstract}

\maketitle


\section{Introduction}
Variational quantum algorithms, and in particular the variational quantum eigensolver (VQE), have become one of the central paradigms for near-term quantum computation~\cite{Peruzzo2014,McClean2016,Kandala2017,Cerezo2021,Tilly2022}. By combining parameterized quantum circuits with classical optimization, VQE provides a hardware-compatible route to approximating ground-state energies in quantum chemistry and other many-body systems. However, noisy intermediate-scale quantum devices remain strongly affected by gate errors, decoherence, readout noise, and finite sampling effects, while full fault-tolerant quantum error correction is still resource demanding for near-term hardware~\cite{Preskill2018,Bharti2022,Devitt2013,LidarBrun2013}. Consequently, substantial effort has been devoted to quantum error mitigation and noise-aware variational strategies, including zero-noise extrapolation, probabilistic error cancellation, and variational error suppression~\cite{Temme2017,LiBenjamin2017,Endo2018,Endo2021}.

At the same time, these approaches usually retain the standard assumption that the variational ansatz itself is generated by unitary gates. Most standard variational ans\"atze therefore prepare only pure states prior to measurement, restricting the accessible variational family to coherent state preparation alone. This raises a more structural question: rather than only mitigating unwanted noise around a unitary ansatz, can one enlarge the variational search space by incorporating controllable nonunitary processes directly into the variational design? From this perspective, open-system dynamics provides a natural route, since it gives rise to mixed states and system transformations beyond standard unitary ans\"atze.

Open quantum systems inevitably interact with their surroundings, leading to decoherence, dissipation, and information loss that are commonly described by quantum dynamical maps and, under the Born-Markov approximation, by the Gorini-Kossakowski-Lindblad-Sudarshan (GKLS) master equation~\cite{BreuerPetruccione2002,RivasHuelga2012,Gorini1976,Lindblad1976}. While such environmental effects are usually regarded as detrimental, reservoir engineering and dissipation engineering show that suitably designed system-environment couplings can instead realize useful nonunitary dynamics, prepare target states, and stabilize quantum resources~\cite{Poyatos1996,Diehl2008,Verstraete2009,Barreiro2011,Schindler2013}. In this sense, dissipation can be treated not only as noise, but also as a controllable variational resource.

On gate-based quantum processors, this idea can be implemented through ancilla-assisted quantum channels. Any completely positive trace-preserving map admits a Stinespring dilation, so that trainable system-ancilla unitaries followed by ancilla trace-out and reset provide a natural circuit-level mechanism for realizing effective nonunitary transformations~\cite{Stinespring1955,Kraus1983,NielsenChuang2010}. Recent work has begun to explore the role of nonunitary dynamics in variational quantum algorithms from several directions. On the one hand, variational methods have been extended beyond standard unitary real-time evolution to settings such as imaginary-time evolution and more general nonunitary simulation~\cite{Yuan2019,McArdle2019,Endo2020}. On the other hand, a number of works have directly considered open-system variational methods, including nonequilibrium steady states and Lindblad dynamics~\cite{yoshioka2020dissipative,watad2024lindblad,chen2024adaptive}. Others have introduced dissipative ingredients such as RESET operations, stochastic gates, or ancilla trace-out directly into the variational ansatz~\cite{ilin2024dvqa}. More task-specific proposals have further considered dissipative eigensolvers for ground-state preparation, Gibbs-state preparation, and lattice gauge theories~\cite{cobos2024noiseaware,cubitt2023dqe}. 

While previous studies have incorporated dissipative ingredients into specific variational settings, a unified ancilla-assisted framework that supports different variational objectives within the same trainable channel architecture is still lacking. In this work, we introduce such a variational dissipative framework for parameterized quantum algorithms. Its basic building block is a system-ancilla unitary followed by ancilla trace-out and reset, which defines a parameterized completely positive trace-preserving map on the main system.
Within this framework, we study two representative tasks. The first is dissipative VQE for ground-state search, in which system-only variational layers are interleaved with trainable dissipative modules. The second is dissipative state recovery, in which the same ancilla-assisted channel structure is trained to improve the fidelity of imperfectly prepared states with respect to a target state. Numerical simulations on representative spin Hamiltonians and target states indicate that this framework can accelerate convergence, enhance noise robustness, and improve final-state fidelity with only a modest number of ancillas.


This paper is organized as follows. In Sec.~\ref{sec:theory}, we review the open-system and channel-theoretic foundations of the method. In Sec.~\ref{sec:framework}, we introduce the variational dissipative framework and define its two task-specific realizations: dissipative VQE and dissipative recovery. In Sec.~\ref{sec:numerics}, we present numerical simulations for ground-state preparation and fidelity recovery. Finally, we summarize the results and discuss possible extensions of trainable dissipation to other variational quantum algorithms.

\section{Theoretical Background}
\label{sec:theory}

In this section, we summarize the theoretical ingredients underlying our framework. We begin with the ground-state search task in the VQE, which serves as the primary variational setting considered in this work. We then introduce the open-system and channel-based description of engineered dissipation, which provides the basis for extending the standard unitary ansatz to a dissipative variational framework.

\subsection{Ground-state search in the VQE framework}

In the standard VQE~\cite{Peruzzo2014,McClean2016,Kandala2017,Tilly2022}, one prepares a parameterized trial state
\begin{equation}
\rho_\theta
=
U(\theta)\ket{0}\bra{0}^{\otimes n}U^\dagger(\theta),
\end{equation}
and minimizes the corresponding energy expectation value
\begin{equation}
E(\theta)=\mathrm{Tr}[H\rho_\theta].
\label{eq:vqe-energy}
\end{equation}
By the variational principle, if $E_0$ denotes the ground-state energy of $H$, then
\begin{equation}
E(\theta)=\mathrm{Tr}[H\rho_\theta]\ge E_0,
\label{eq:variational-bound}
\end{equation}
so the variational energy provides an upper bound on the true ground-state energy.

In the standard unitary formulation of VQE, the variational ansatz prepares a family of pure states. In this setting, the accessible variational manifold is restricted to coherent state preparation alone. Allowing dissipative operations provides a natural way to enlarge this variational family, since the resulting trial states may be mixed and can be described by parameterized completely positive trace-preserving (CPTP) maps~\cite{Yuan2019,Endo2020}. Accordingly, the trial state may be written more generally as
\begin{equation}
\rho_\theta
=
\mathcal{D}_\theta\!\left(\ket{0}\bra{0}^{\otimes n}\right),
\label{eq:diss-vqe-state}
\end{equation}
where $\mathcal{D}_\theta$ denotes a parameterized dissipative channel. The optimization objective remains the energy expectation value in Eq.~\eqref{eq:vqe-energy}, but the admissible state manifold is enlarged beyond purely unitary state preparation. This observation provides the basic motivation for incorporating engineered dissipation into the VQE framework.

\subsection{Open-system dynamics and engineered dissipation}

An open quantum system does not evolve in isolation, but interacts with external degrees of freedom that act as an environment. As a result, the reduced dynamics of the system is generally nonunitary: coherence can be lost, entropy can be exchanged with the environment, and information can flow out of the system. Under the Born-Markov approximation, this reduced dynamics is described by the GKLS master equation~\cite{Gorini1976,Lindblad1976,BreuerPetruccione2002,RivasHuelga2012}
\begin{equation}
\frac{d\rho_S}{dt}
=
-i[H_S,\rho_S]
+
\sum_k \gamma_k
\left(
L_k \rho_S L_k^\dagger
-
\frac{1}{2}\{L_k^\dagger L_k,\rho_S\}
\right),
\label{eq:gkls}
\end{equation}
where $H_S$ is the system Hamiltonian, and $\{L_k,\gamma_k\}$ are the Lindblad jump operators and their corresponding dissipation rates.

For our purposes, the main point is not only that open-system dynamics is nonunitary, but that it can also be engineered. If the coupling between the system and auxiliary degrees of freedom is controllable, then dissipation need not be viewed only as an unwanted effect. Instead, it can be used to realize designed state transformations that are inaccessible to purely unitary evolution on the system alone.

On a gate-based quantum computer, this idea can be formulated in a simple and general way through ancilla-assisted quantum channels. Any completely positive trace-preserving (CPTP) map $\mathcal{E}$ acting on the system can be written in the Stinespring form~\cite{Stinespring1955,Kraus1983,NielsenChuang2010}
\begin{equation}
\mathcal{E}(\rho_S)
=
\mathrm{Tr}_A\!\left[
V_{SA}
\left(
\rho_S \otimes \ket{0}\bra{0}_A
\right)
V_{SA}^\dagger
\right],
\label{eq:stinespring}
\end{equation}
where $V_{SA}$ is a joint unitary acting on the system and an ancilla initialized in a reference state. In this representation, the nonunitary dynamics of the system arises from a larger unitary evolution followed by tracing out the ancilla.

Projecting the ancilla onto an orthonormal basis $\{\ket{i}_A\}$ yields the corresponding Kraus representation~\cite{Kraus1983,NielsenChuang2010}
\begin{equation}
K_i = \bra{i}_A V_{SA} \ket{0}_A,
\qquad
\mathcal{E}(\rho_S)=\sum_i K_i \rho_S K_i^\dagger.
\label{eq:kraus}
\end{equation}
Thus, the ancilla-assisted picture, the Stinespring dilation, and the Kraus-operator description are simply different ways of expressing the same open-system transformation.

The same channel picture also connects dissipation with recovery. If the ancilla is traced out, one obtains an unconditional dissipative channel. If it is measured and used for conditional correction, the same enlarged evolution can instead be interpreted as a recovery process. In this sense, dissipation and recovery can be viewed within a common ancilla-assisted framework, differing mainly in how the auxiliary degrees of freedom are processed after the joint evolution~\cite{WisemanMilburn2010}.


The above considerations provide the conceptual basis for our variational dissipative construction: by augmenting standard unitary state preparation with trainable ancilla-assisted dissipation, one can extend the usual VQE ansatz to a more general open-system variational framework.

\section{Variational Dissipative Framework}
\label{sec:framework}

We now introduce the variational dissipative framework used in this work. The central idea is to combine system-only variational layers with ancilla-assisted dissipative modules, so that the circuit includes both coherent evolution and trainable nonunitary state transformations. We first describe the common circuit structure and then present its two task-specific realizations: ground-state search and state recovery.

\subsection{Overall architecture}

\begin{figure*}[ht]
    \centering
    \includegraphics[width=0.8\linewidth]{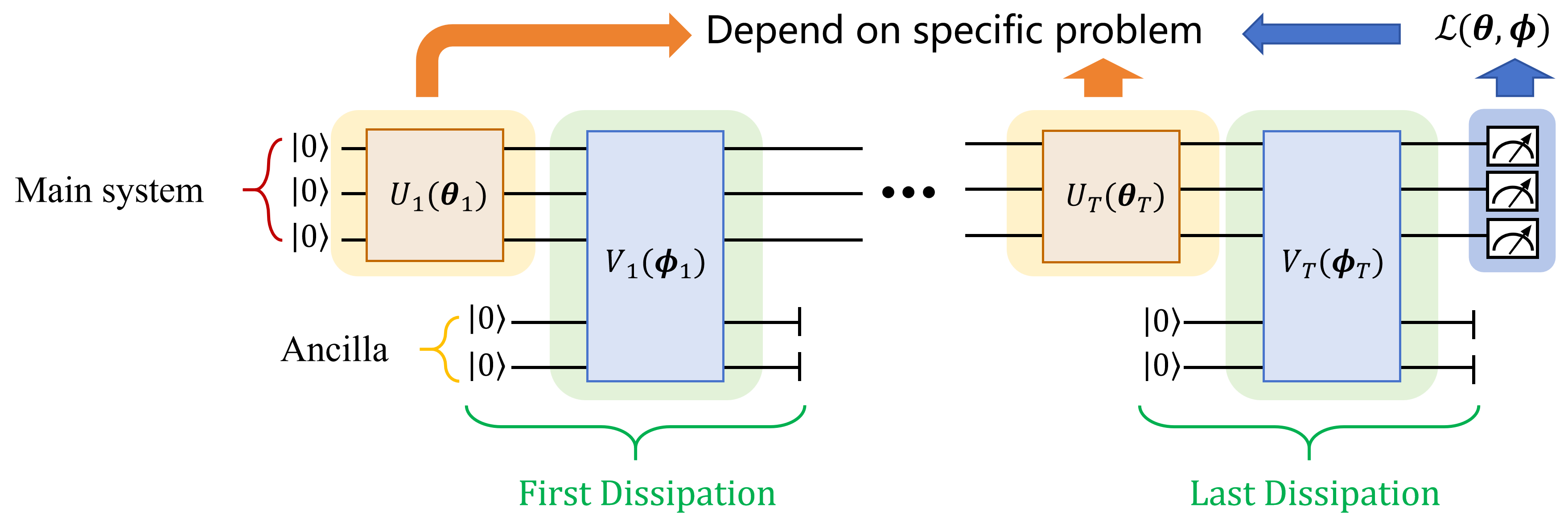}
    \caption{General architecture of the variational dissipative framework. The circuit combines a system-only variational module with an ancilla-assisted dissipative module.}
    \label{NS0}
\end{figure*}

Fig.~\ref{NS0} illustrates the overall architecture of the proposed framework. The circuit consists of a main system with $n$ qubits, which carries the computational state, and an ancilla register with $m$ qubits, which provides auxiliary degrees of freedom for engineered dissipation. The framework contains two complementary modules. The first is a system-only variational block $U_{\mathrm{vqe}}$, whose role is to explore a task-dependent variational manifold through coherent unitary evolution. The second is a variational dissipative block $U_{\mathrm{dis}}$, whose role is to couple the system to the ancillas and induce an effective nonunitary transformation after the ancillas are traced out and reset.

These two modules define the common structure of the framework. In the ground-state-search setting, they are interleaved and optimized jointly through an energy objective. In the state-recovery setting, the same dissipative module is retained, while the system-only variational block is removed and the objective is replaced by a fidelity-based target. The main point is that the same ancilla-assisted dissipative mechanism can be used across different tasks, even though the system-side transformation and the optimization objective are different.

\subsection{Ancilla-assisted trainable dissipation}

Each dissipative block is modeled as a parameterized Stinespring dilation of a trainable quantum channel,
\begin{equation}
\mathcal{D}_{\theta}(\rho_S)
=
\mathrm{Tr}_A\left[
V_{SA}(\theta)
\left(
\rho_S \otimes \ket{0}\bra{0}_A^{\otimes m}
\right)
V_{SA}^\dagger(\theta)
\right],
\label{eq:var-diss-channel}
\end{equation}
where the ancilla qubits are initialized in $\ket{0}_A^{\otimes m}$, interact jointly with the system through a parameterized unitary $V_{SA}(\theta)$, and are then traced out. The trainable parameters are therefore carried by the system-ancilla interaction itself. Since the ancillas are reset before each application, every dissipative block defines a completely positive trace-preserving map acting effectively on the system alone. The resulting variational circuit should thus be viewed more generally as a sequence of coherent and dissipative transformations.

This construction provides a simple way to introduce trainable nonunitarity into variational quantum protocols. In contrast to a fixed physical noise channel, the dissipative map in Eq.~\eqref{eq:var-diss-channel} is optimized together with the other task-dependent circuit parameters. The system-ancilla interaction therefore acts as a controllable channel layer that reshapes the reduced state of the system during optimization.

In our numerical implementation, the dissipative module is composed of arbitrary single-qubit rotations on the involved qubits together with system-ancilla $i$SWAP couplings~\cite{Schuch2003,Barends2019}. We use a fully connected system-ancilla pattern, where each system qubit is coupled to each ancilla qubit, leading to $nm$ system-ancilla $i$SWAP gates in one dissipative round for a system of $n$ qubits and an ancilla register of $m$ qubits. Other choices of connectivity or coupling gates can also be incorporated within the same framework.

\subsection{Ground-state-search instantiation}

For the ground-state-search task, the variational protocol alternates between coherent system-side transformations and ancilla-assisted dissipative updates. The goal is to minimize the energy expectation value of a target Hamiltonian through a VQE-type optimization.

The dissipative VQE protocol can be formulated recursively. Let $\rho_{\mathrm{in}}$ denote the initial system state. In the present implementation, we take $\rho_{\mathrm{in}}=\ket{0}\bra{0}^{\otimes n}$, so that
\begin{equation}
\rho_S^{(0)}
=
U_{\mathrm{vqe}}^{(0)}
\rho_{\mathrm{in}}
U_{\mathrm{vqe}}^{(0)\dagger}.
\label{VQE2}
\end{equation}
Then, for $t=1,\dots,T$, each dissipative round is defined by
\begin{equation}
\rho_S^{(t)}
=
\mathrm{Tr}_A
\left[
U_{\mathrm{dis}}^{(t)}
\left(
U_{\mathrm{vqe}}^{(t)}
\rho_S^{(t-1)}
U_{\mathrm{vqe}}^{(t)\dagger}
\otimes
|0\rangle\langle 0|_A^{\otimes m}
\right)
U_{\mathrm{dis}}^{(t)\dagger}
\right],
\label{VQE3}
\end{equation}
where $\mathrm{Tr}_A$ denotes the partial trace over the ancilla register. After each round, the ancillas are traced out, reset to $|0\rangle^{\otimes m}$, and reintroduced in the next step. The loss function is the final energy expectation value
\begin{equation}
\mathcal{L}_{\mathrm{VQE}}
=
\mathrm{Tr}\!\left[ H \rho_S^{(T)} \right].
\label{VQE4}
\end{equation}

In this instantiation, both the system-only variational layers and the dissipative blocks are parameterized and optimized jointly. The protocol therefore alternates between coherent exploration of the variational manifold and dissipative reshaping of the system state within a single optimization loop. We also provide the analysis about the performance between original VQE and dissipative VQE in Appendix.~\ref{sec:analysis}.

\subsection{State-recovery instantiation}

For the state-recovery task, the goal is no longer to minimize an energy expectation value, but to improve the fidelity of an imperfect input state with respect to a target state. To this end, we use the same ancilla-assisted dissipative block as in the VQE setting, but remove the system-only variational layers. The recovery protocol therefore consists of repeated dissipative updates applied directly to the input state, so that the reduced state of the main system is progressively reshaped toward the target state through a trainable system-ancilla channel.

For the recovery task considered here, the target state is taken to be a pure state
\begin{equation}
\rho_{\mathrm{tar}}
=
\ket{\psi_{\mathrm{tar}}}\bra{\psi_{\mathrm{tar}}},
\label{Fidelity1}
\end{equation}
and let $\rho_{\mathrm{in}}$ denote the corresponding imperfect input state. The purpose of the recovery protocol is to transform $\rho_{\mathrm{in}}$ into a final reduced state that is as close as possible to $\rho_{\mathrm{tar}}$.

We initialize the protocol with
\begin{equation}
\rho_0 = \rho_{\mathrm{in}},
\label{Fidelity2a}
\end{equation}
and then define the recovery dynamics recursively by
\begin{equation}
\rho_t
=
\mathrm{Tr}_A
\left[
U_{\mathrm{dis}}^{(t)}
\left(
\rho_{t-1}
\otimes
|0\rangle\langle 0|_A^{\otimes m}
\right)
U_{\mathrm{dis}}^{(t)\dagger}
\right],
\label{Fidelity2}
\end{equation}
where $t=1,\dots,T$ and the ancillas are traced out and reset after each round.

The performance metric is the fidelity with the target state,
\begin{equation}
F_T
=
F\left(\rho_T,\rho_{\mathrm{tar}}\right)
=
\left(
\mathrm{Tr}
\sqrt{
\sqrt{\rho_{\mathrm{tar}}}
\, \rho_T
\sqrt{\rho_{\mathrm{tar}}}
}
\right)^2,
\label{Fidelity3}
\end{equation}
which reduces to
\begin{equation}
F_T
=
\mathrm{Tr}
\left(
\rho_T \rho_{\mathrm{tar}}
\right)
\label{Fidelity4}
\end{equation}
because $\rho_{\mathrm{tar}}$ is pure. Equivalently, one may minimize the loss
\begin{equation}
\mathcal{L}_{\mathrm{rec}}
=
1-F_T.
\label{Fidelity5}
\end{equation}

In this case, the dissipative block is again parameterized and trained, but now against a fidelity-based objective rather than an energy-based one. The recovery task therefore shares the same ancilla-assisted dissipative structure as the VQE task, while differing in both the system-side transformation and the optimization target.

\section{Numerical simulation}
\label{sec:numerics}

The variational dissipative framework developed above suggests that trainable system-ancilla dissipation can provide additional channel degrees of freedom for reshaping the system state. In the ground-state-search setting, this mechanism is expected to accelerate the variational energy descent, improve the final convergence accuracy, and enhance robustness against noise. In the state-recovery setting, the same mechanism is expected to improve the fidelity of imperfect input states with respect to the target state.

In this section, we organize the numerical results around three main questions. First, does the framework improve task performance in practice? Second, does this improvement remain robust and general across different settings? Third, how does the performance depend on the dissipative resources, including the number of dissipative rounds and ancillas?

For the ground-state-search task, we benchmark the dissipative VQE on the open-boundary nearest-neighbor spin Hamiltonian
\begin{equation}
    \begin{split}
        H
        =&
        \sum_{i=1}^{n-1}
        \left(
        J_x X_i X_{i+1}
        +
        J_y Y_i Y_{i+1}
        +
        J_z Z_i Z_{i+1}
        \right)
        \\
        &+
        \sum_{i=1}^{n}
        \left(
        h_x X_i
        +
        h_y Y_i
        +
        h_z Z_i
        \right),
    \end{split}
\label{VQE1}
\end{equation}
where the first summation describes nearest-neighbor interaction terms and the second summation represents local fields. Here $J_x,J_y,J_z$ and $h_x,h_y,h_z$ are real coefficients that specify the model. Unless otherwise specified, the exact ground-state energies used as references are obtained by exact diagonalization of the corresponding Hamiltonian matrices.

For the state-recovery task, we consider three representative target states with different preparation structures. The first is the $W$ state~\cite{Dur2000}
\begin{equation}
\ket{\psi_{W}}
=
\ket{W_n}
=
\frac{1}{\sqrt{n}}
\sum_{i=1}^{n}
\ket{0\cdots 010\cdots 0},
\label{Fidelity6}
\end{equation}
the second is the uniform superposition state
\begin{equation}
\ket{\psi_{+}}
=
\ket{+}^{\otimes n}
=
\frac{1}{\sqrt{2^n}}
\sum_{x\in\{0,1\}^n}
\ket{x},
\label{Fidelity7}
\end{equation}
and the third is a dressed cluster state
\begin{equation}
\ket{\psi_{\mathrm{DC}}}
=
\prod_{d=1}^{D}
\left(
\prod_{i=1}^{n-1} CZ_{i,i+1}
\right)
\left(
\bigotimes_{i=1}^{n} R_y(\alpha_i^{(d)})
\right)
\ket{+}^{\otimes n},
\label{Fidelity8}
\end{equation}
with $D=3$ and $\alpha_i^{(d)}=\pi/4$ for all $i$ and $d$~\cite{Briegel2001,Raussendorf2001,Hein2004}.

\subsection{The framework improves both optimization and recovery performance}

We first examine whether the proposed framework improves task performance in practice. For the ground-state-search task, we begin with the representative choice $J_x=1.0$ and $h_z=0.3$ in Eq.~\eqref{VQE1}, with all other coefficients set to zero, corresponding to an $XX$-type model in a transverse $Z$ field.

\begin{figure}[t]
    \centering
    \includegraphics[width=\linewidth]{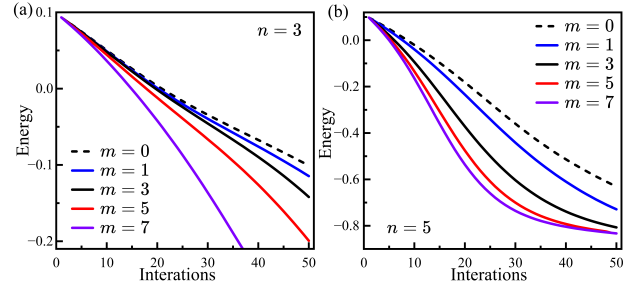}
    \caption{Training curves of the dissipative VQE with different numbers of ancillas. Panels (a) and (b) correspond to different main-system sizes. As the number of ancillas increases, the energy decreases faster during training and the final converged energy becomes lower, demonstrating improved convergence speed and variational accuracy.}
    \label{NS2_VQEAncilla}
\end{figure}

Fig.~\ref{NS2_VQEAncilla} shows that the training performance improves systematically with the ancilla number. In this test, the learning rate is set to $lr=0.2$, the number of dissipative rounds is fixed at $T=1$, and the training is performed under depolarizing noise with strength $p=0.1$. In both panels, increasing $m$ leads to a faster early-stage energy descent and a lower final converged energy compared with the ancilla-free baseline. For $n=3$ in panel (a), the curves separate gradually as $m$ increases, while for $n=5$ in panel (b), the same trend becomes much more pronounced, with larger ancilla numbers producing clearly steeper descent and substantially lower final energies. These results show that the dissipative construction improves both convergence speed and final variational accuracy, and that this improvement generally becomes stronger as the ancilla register is enlarged.

\begin{figure}[t]
    \centering
    \includegraphics[width=\linewidth]{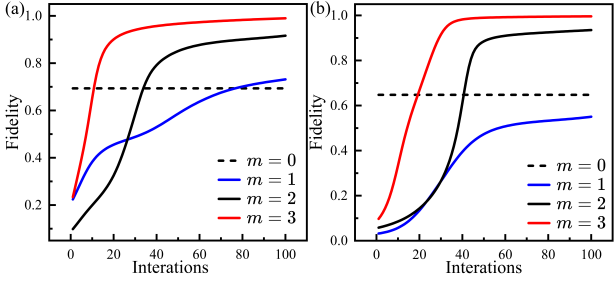}
    \caption{Fidelity recovery performance for different numbers of ancillas. Introducing dissipative ancillas significantly improves the fidelity beyond the initial input value, and increasing the number of ancillas generally leads to both faster recovery and a higher final fidelity.}
    \label{NS8_FidelityAncilla}
\end{figure}
The same qualitative gain appears in the state-recovery task. Fig.~\ref{NS8_FidelityAncilla} shows a clear ancilla-number dependence in the recovery dynamics. Once dissipative ancillas are introduced, the fidelity rises well above the input-state baseline, and larger values of $m$ generally produce both faster recovery and higher final fidelity. These results show that the dissipative construction is also beneficial in the recovery setting, and that this advantage becomes more pronounced as the ancilla register is enlarged.

Taken together, these results show that the proposed framework is useful in both of its main realizations: it improves the optimization trajectory in dissipative VQE and increases the achievable fidelity in dissipative state recovery.

\subsection{The framework remains robust and general}

We next examine whether these improvements remain visible beyond a single benchmark setting.

\begin{figure}[t]
    \centering
    \includegraphics[width=\linewidth]{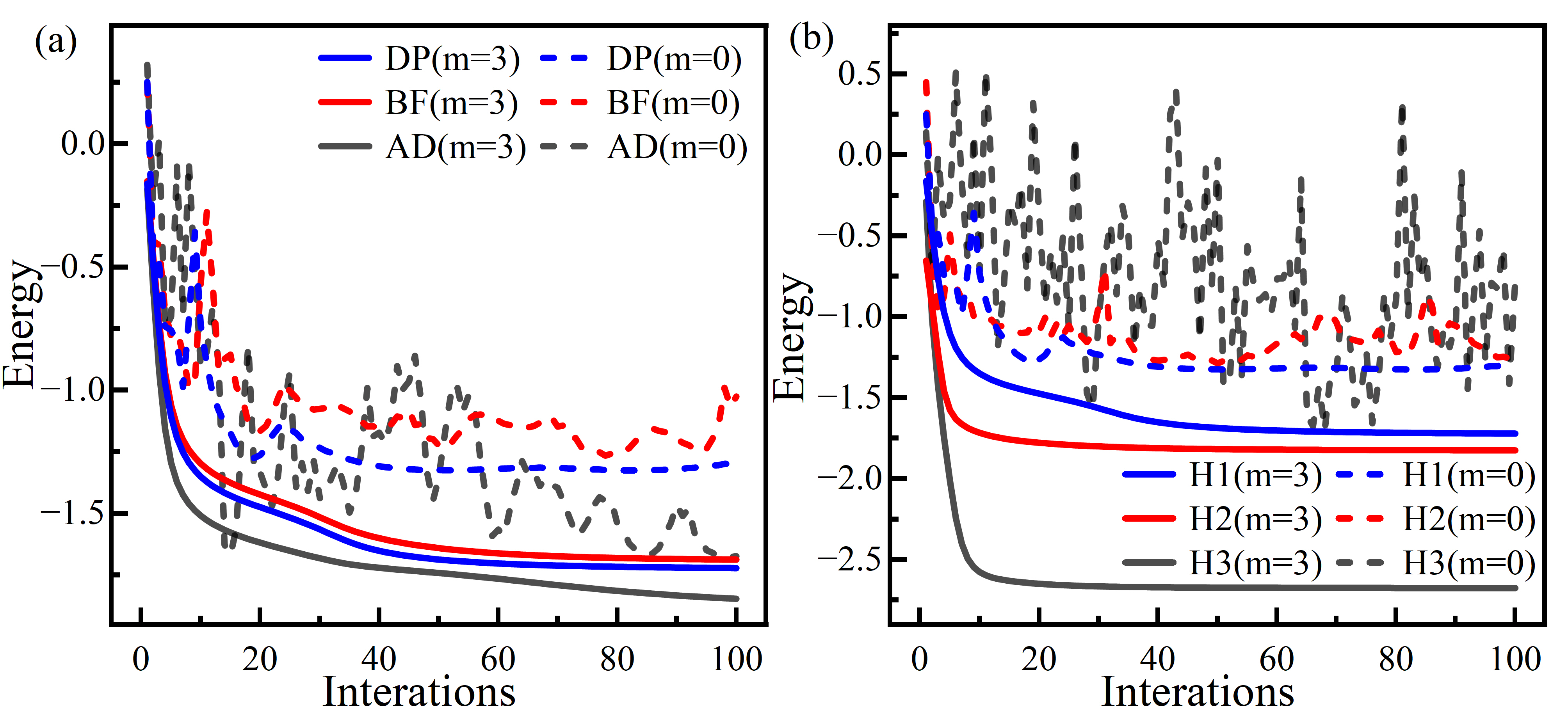}
    \caption{
    Robustness of the dissipative VQE across different noise channels and Hamiltonians.
    (a) Comparison under amplitude-damping, bit-flip, and depolarizing noise in the fully noisy setting.
    (b) Comparison for three representative Hamiltonians: an $XX$-type model in a transverse field, a transverse-field Ising model, and an $XXZ$-type model.
    In both cases, the dissipative VQE converges faster and reaches lower final energies than the ancilla-free baseline.
    }
    \label{NS5_VQENoiseTypesHamiltonians}
\end{figure}

For the VQE task, Fig.~\ref{NS5_VQENoiseTypesHamiltonians}(a) shows that the dissipative construction remains advantageous across qualitatively different noise channels. In the fully noisy setting considered here, noise acts not only on the main-system variational circuit but also on the dissipative blocks themselves. Under Amplitude Damping (AD), Bit-Flip (BF), and Depolarizing (DP) noise, the dissipative VQE still converges faster, fluctuates less, and reaches lower final energies than the ancilla-free baseline. This indicates that the optimized system-ancilla channel provides a net improvement even when the dissipative circuit itself is noisy. A complementary scan over the noise strength is shown in Appendix Fig.~\ref{NS3_VQENoiseStrength}, which leads to the same conclusion: the dissipative VQE maintains a smaller gap to the exact ground-state value as the noise strength increases.

The same qualitative advantage also persists across representative Hamiltonians. We consider
\begin{align}
    H_1 &=
    \sum_{i=1}^{n-1} X_i X_{i+1}
    +
    0.3\sum_{i=1}^{n} Z_i,
    \\
    H_2 &=
    \sum_{i=1}^{n-1} Z_i Z_{i+1}
    +
    0.3\sum_{i=1}^{n} X_i,
    \\
    H_3 &=
    \sum_{i=1}^{n-1} X_i X_{i+1}
    +
    \sum_{i=1}^{n-1} Y_i Y_{i+1}
    +
    0.3\sum_{i=1}^{n-1} Z_i Z_{i+1},
\label{VQE5}
\end{align}
where $H_1$ is an $XX$-type model in a transverse $Z$ field, $H_2$ is the transverse-field Ising model, and $H_3$ is an $XXZ$-type model~\cite{Lieb1961,Pfeuty1970,Sachdev2011,Baxter1982}. The data are obtained under the same fully noisy setting, with depolarizing noise of strength $p=0.01$. 

As shown in Fig.~\ref{NS5_VQENoiseTypesHamiltonians}(b), the dissipative VQE consistently converges faster, behaves more smoothly, and reaches final energies closer to the exact ground-state values than the ancilla-free baseline for all three Hamiltonians. Since these Hamiltonians have different interaction structures and therefore different minimum energy values, the consistent improvement suggests that the dissipative block is not only fitted to one specific energy landscape. 

The smoother training behavior observed in Fig.~\ref{NS5_VQENoiseTypesHamiltonians} can be understood from the optimization structure of the dissipative ansatz. Compared with the oscillation curve without ancillas, the dissipative model contains more trainable parameters, so under the same learning rate the parameter update is distributed over a larger parameter set and the energy change between neighboring iterations can become less abrupt. However, the improvement is not only a consequence of adding more parameters. Through the system-ancilla coupling followed by ancilla trace-out and reset, the dissipative block acts as a trainable channel on the main system, rather than merely as a deeper unitary circuit. This gives the optimizer an additional way to modify the system state during training. Therefore, the smoother curves should be viewed as a combined effect of the enlarged ansatz and the dissipative channel structure, while the lower final energies show that this dissipative channel is useful for finding a lower energy value.

The recovery task shows robustness across different preparation-noise channels. Fig.~\ref{NS10_FidelityNoiseStates}(a) focuses on a fixed dressed-cluster target state and compares different preparation-noise channels. The dissipative recovery protocol improves the final fidelity over the input-state baseline under DP, BP, and AD noise. This behavior is consistent with the intended role of the variational dissipative block: for each noise condition, the system-ancilla interaction is optimized to transform the noisy input state toward the same target state. The result therefore indicates that the protocol is not restricted to a single preparation-noise model.

\begin{figure}[ht]
    \centering
    \includegraphics[width=\linewidth]{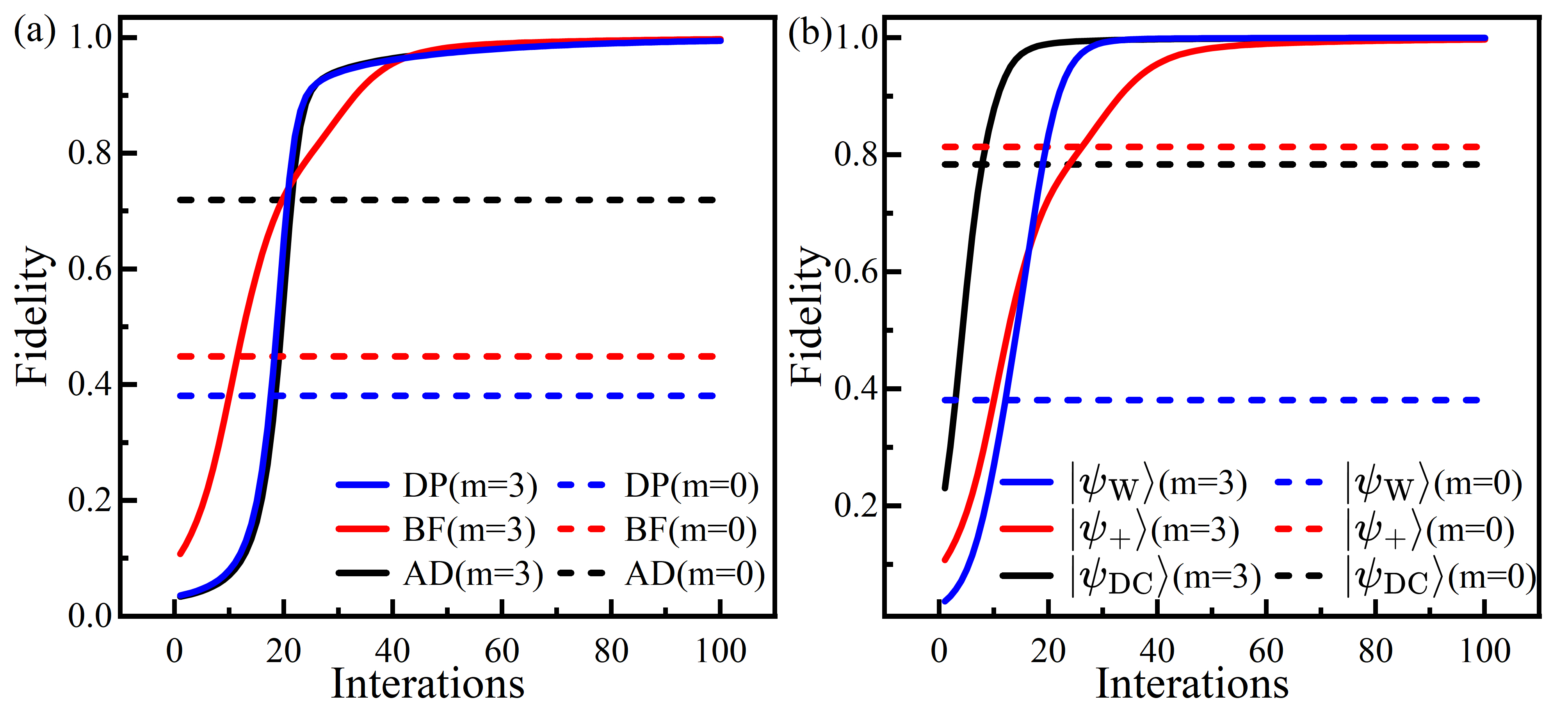}
    \caption{
    Fidelity recovery under different noise conditions and target states.Dashed lines denote the noisy input fidelities without dissipative recovery, while solid lines denote the results with variational dissipative recovery.
    (a) Recovery of a dressed cluster target state under depolarizing, bit-flip, and amplitude-damping noise.
    (b) Recovery performance for three representative target states: the $W$ state $\ket{\psi_W}$, the uniform superposition state $\ket{\psi_+}$, and the dressed cluster state $\ket{\psi_{\mathrm{DC}}}$.
    Unless otherwise specified, the simulations use $n=3$ system qubits, $m=3$ ancillas, learning rate $lr=0.8$, dissipative rounds $T=3$, and noise strength $p=0.1$.
    }
    \label{NS10_FidelityNoiseStates}
\end{figure}

The state-recovery task also shows the same kind of generality across target states. In the experiments of Fig.~\ref{NS10_FidelityNoiseStates}(b), we fix $m=3$, $T=3$, and $lr=0.8$, and prepare the input states under DP noise with strength $p=0.1$. For all three target states, the dissipative recovery model restores the final fidelity to above $99\%$ within 100 optimization iterations. These results show that the same dissipative recovery mechanism remains effective across qualitatively different target-state structures.

These results indicate that the improvement provided by trainable dissipation is not tied to a single benchmark example. The common mechanism behind these tests is that the ancilla-assisted block gives the variational protocol a trainable state-transformation step beyond system-only unitary evolution. In the VQE task, this additional channel-type update helps the optimization reach lower-energy states across different noisy settings and Hamiltonian landscapes. In the recovery task, the same structural idea allows the protocol to transform noisy input states toward the desired target under different preparation errors and for different target states. Therefore, the robustness observed here should be understood not as identical performance in every case, but as a consistent advantage of adding a trainable dissipative transformation to the variational workflow.

\subsection{The framework exhibits structured resource behavior}

Finally, we examine how the performance gain depends on the dissipative resources themselves. The numerical results suggest a structured resource behavior rather than an indefinite improvement with increasing dissipative complexity: moderate dissipative depth is typically sufficient, and the benefit of additional ancillas eventually saturates.

\begin{figure}[h]
    \centering
    \includegraphics[width=\linewidth]{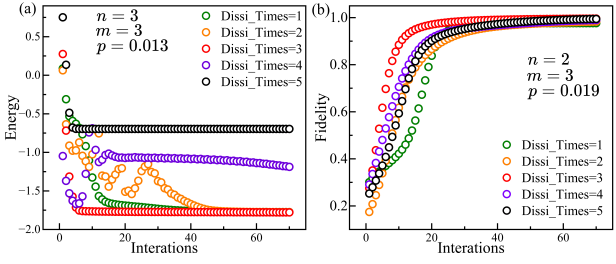}
    \caption{
    Dependence of the performance on the number of dissipative rounds $T$.
    (a) Energy optimization in the dissipative VQE task.
    (b) Fidelity recovery in the dissipative state-recovery task.
    The physical and hyperparameter settings are indicated in each panel and are chosen separately for the two tasks to illustrate the resource behavior in different settings.
    In both tasks, a moderate number of dissipative rounds, typically $T=3$, provides a good balance between convergence speed, final performance, and training overhead.
    }
    \label{NS4_Times_VQE_Fidelity}
\end{figure}

For the resource dependence on the number of dissipative rounds, Fig.~\ref{NS4_Times_VQE_Fidelity} compares representative scans in the VQE and state-recovery tasks. The two panels use different physical parameters, as indicated in the figure, because the goal is not to compare the two tasks directly but to show that a similar resource trend appears in different settings.

In the VQE task, Fig.~\ref{NS4_Times_VQE_Fidelity}(a) shows a non-monotonic dependence on $T$: too few dissipative rounds provide limited state reshaping, whereas too many rounds increase the circuit depth and the number of trainable parameters, which can introduce additional optimization overhead. In this example, $T=3$ gives the best overall energy convergence.

A similar tradeoff appears in the fidelity-recovery task in Fig.~\ref{NS4_Times_VQE_Fidelity}(b). Larger values such as $T=4$ or $T=5$ can slightly exceed the $T=3$ curve after about 55 iterations, but they require a substantially longer training time and show weaker early-stage convergence. Therefore, the small late-stage improvement does not necessarily justify the additional cost. Taken together, these results indicate that a moderate number of dissipative rounds already captures most of the available benefit. For this reason, we use $T=3$ in most of the remaining experiments.

\begin{figure}[ht]
    \centering
    \includegraphics[width=\linewidth]{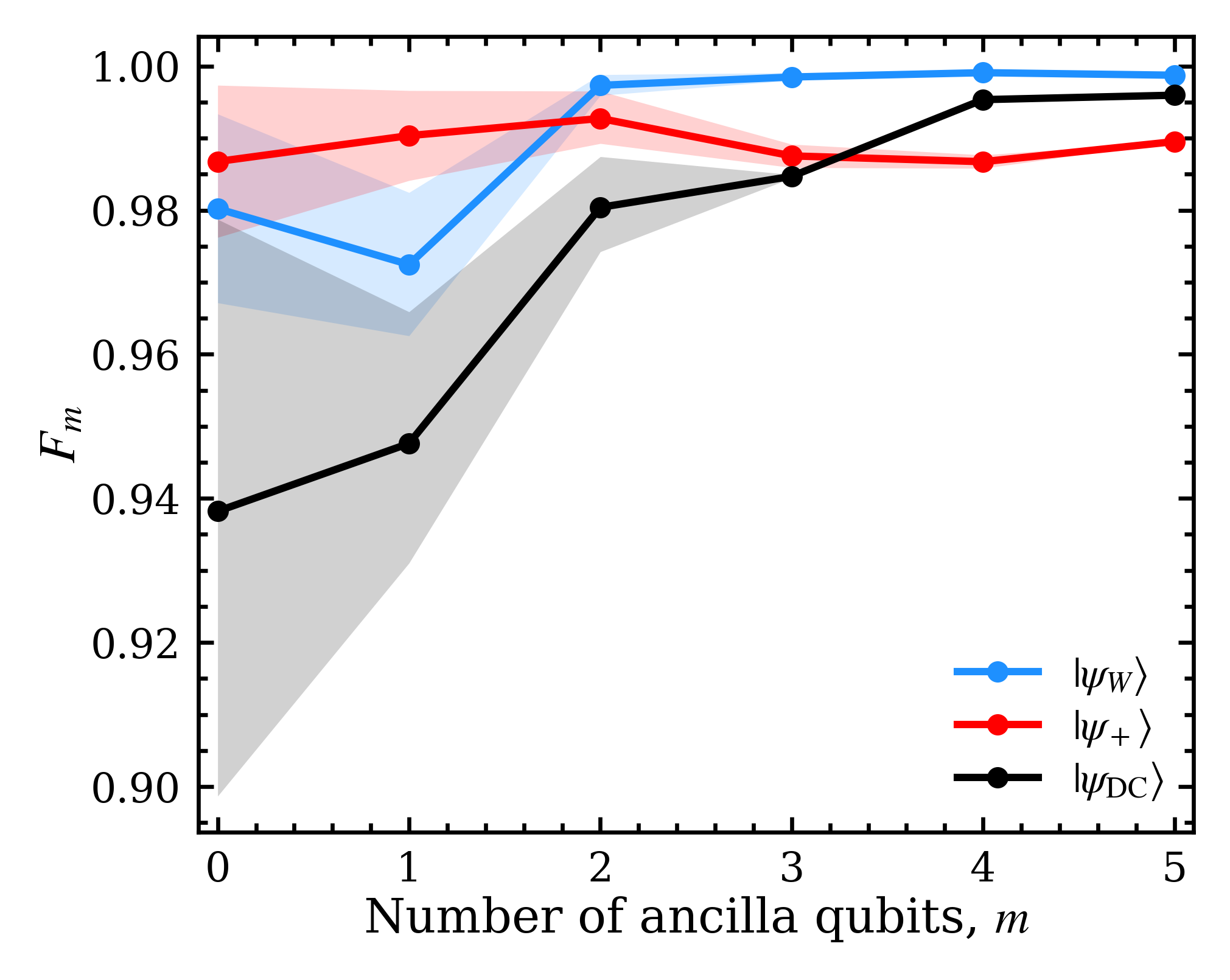}
    \caption{Final fidelity $F_m$ versus the number of ancillas for different noise strengths, showing state-dependent ancilla saturation and the emergence of a robustness plateau. For each target state, the solid line denotes the central trend, while the shaded region indicates the envelope bounded by the highest- and lowest-fidelity curves over different noise strengths.}
    \label{NS11_FidelitySaturation}
\end{figure}


A more detailed resource effect appears in the state-recovery task when varying the number of ancillas.
Fig.~\ref{NS11_FidelitySaturation} shows that the final fidelity $F_m$ approaches a plateau as the number of ancillas increases. More importantly, the separation between curves with different noise strengths becomes much smaller after the plateau is reached. This indicates not only an ancilla-saturation effect, but also a reduction of noise sensitivity.

This behavior can be understood from the channel perspective. With too few ancillas, the dissipative module has only a limited number of effective Kraus degrees of freedom, and therefore can correct or suppress only a restricted subset of the noise-induced error components. In this regime, stronger preparation noise leaves a larger residual component outside the recoverable subspace, so the final fidelity remains sensitive to the noise strength. Increasing the number of ancillas enlarges the effective channel space and allows the optimized system-ancilla interaction to capture more of the dominant error modes. Once these dominant recoverable modes are covered, the learned channel can map noisy inputs generated at different noise strengths toward a similar high-fidelity region near the target state. The remaining errors are either weak or incompatible with the fixed dissipative ansatz, so further ancillas provide only marginal additional improvement. We refer to this empirical crossover as \textit{ancilla saturation}. 


The effective saturation value is state dependent. For the three target states considered in Fig.~\ref{NS11_FidelitySaturation}, the approximate saturation values are $m=2$ for $\ket{\psi_W}$, $m=5$ for $\ket{\psi_{+}}$, and $m=3$ for $\ket{\psi_{\mathrm{DC}}}$. This suggests that the required ancilla resources are not determined solely by the amount of entanglement in the target state. Instead, they depend on the combined structure of the preparation circuit, the noise-induced error modes, and the expressive compatibility between these error modes and the chosen dissipative ansatz.

A complementary view is provided by the fidelity gain shown in Appendix Fig.~\ref{NS12_FidelityThreshold} for the $W$-state recovery task. Once $m\geq 2$, the gain curves for different ancilla numbers almost overlap, showing that the dominant recoverable error modes have already been captured and that further ancillas do not provide appreciable additional benefit.

These results suggest that the benefit of trainable dissipation is not obtained by arbitrarily increasing circuit complexity. Instead, the numerical behavior exhibits effective resource thresholds and saturation, indicating that the useful dissipative degrees of freedom are finite and task dependent in the present framework.

\subsection{Summary of the numerical results}

Across all numerical experiments, a common pattern emerges: the ancilla-assisted dissipative module acts as an effective trainable channel layer that reshapes the reduced system state in a task-dependent but consistently useful way. In the ground-state-search setting, this reshaping accelerates the variational descent, improves the final energy accuracy, and enhances robustness under noise. In the state-recovery setting, it improves the overlap with the target state and provides a controllable route for suppressing preparation errors. The main numerical message is therefore not only that the framework works in two separate tasks, but that the same dissipative mechanism can serve as a reusable variational resource across different objectives.

\section{Conclusion}

In this work, we introduced a variational dissipative framework in which system-only variational layers are combined with ancilla-assisted trainable dissipation. The central idea is to use a common dissipative block, implemented through parameterized system-ancilla couplings followed by ancilla trace-out and reset, so that coherent variational evolution and controllable nonunitary state transformations can be incorporated within the same circuit architecture.

Within this framework, we studied two representative tasks: ground-state search and state recovery. In the ground-state-search setting, the dissipative construction improves the optimization trajectory, lowers the final converged energy, and remains effective across different noise conditions and benchmark Hamiltonians. In the state-recovery setting, the same dissipative mechanism improves the fidelity of imperfect input states with respect to the target state and remains effective across target states with different structures. Taken together, these results show that the ancilla-assisted dissipative block is not tied to a single objective, but can serve as a reusable trainable component across different variational tasks.

Our numerical results also reveal a more structured picture of the dissipative resources themselves. The gain from dissipation is not obtained simply by increasing circuit complexity without bound. Instead, the framework already performs well with a moderate number of dissipative rounds, while the ancilla requirement in the recovery task exhibits a clear task-dependent saturation behavior. This suggests that the useful dissipative degrees of freedom are finite and can be meaningfully optimized within a practical variational setting.

Overall, our results support a broader view of engineered dissipation in variational quantum computing. Rather than treating nonunitary dynamics only as an unwanted effect, one can incorporate ancilla-assisted dissipation directly into the variational design and use it as a trainable resource. Looking ahead, it will be interesting to understand more systematically how this framework scales with system size, how its resource requirements depend on the target task and error model, and how similar dissipative modules can be incorporated into other variational quantum algorithms.

\section*{Acknowledgments}
We thank Ying Li for helpful discussions. This work is supported by 
Quantum Science and Technology-National Science and Technology Major Project (2023ZD0300200), Beijing Natural Science Foundation Z250004, NSAF (Grant No.~U2330201), the National Natural Science Foundation of China Grant (No.~12361161602), Beijing Science and Technology Planning Project (Grant No.~Z25110100810000), and the High-performance Computing Platform of Peking University.

\bibliographystyle{apsrev4-2}
\bibliography{references}

\clearpage
\newpage
\onecolumngrid
\appendix

\section{Additional numerical results}
\label{app:numerics}
This section collects supplementary numerical results that support the main conclusions of Sec.~\ref{sec:numerics}. These figures do not introduce new qualitative claims, but provide additional evidence for the robustness and resource behavior of the proposed variational dissipative framework.

\subsection{Dissipative models}

Fig.~\ref{NS1_DVQE} shows the circuit-level realization of the dissipative VQE model used in our simulations. The yellow blocks denote the system-only variational layers, while the green blocks denote the dissipative modules implemented by system-ancilla couplings followed by ancilla trace-out and reset. In the multi-round construction, a system-only variational block is inserted before each dissipative block, so that the circuit can alternate between coherent exploration of the variational manifold and dissipative reshaping of the system state.
\begin{figure}[H]
    \centering
    \includegraphics[width=.9\linewidth]{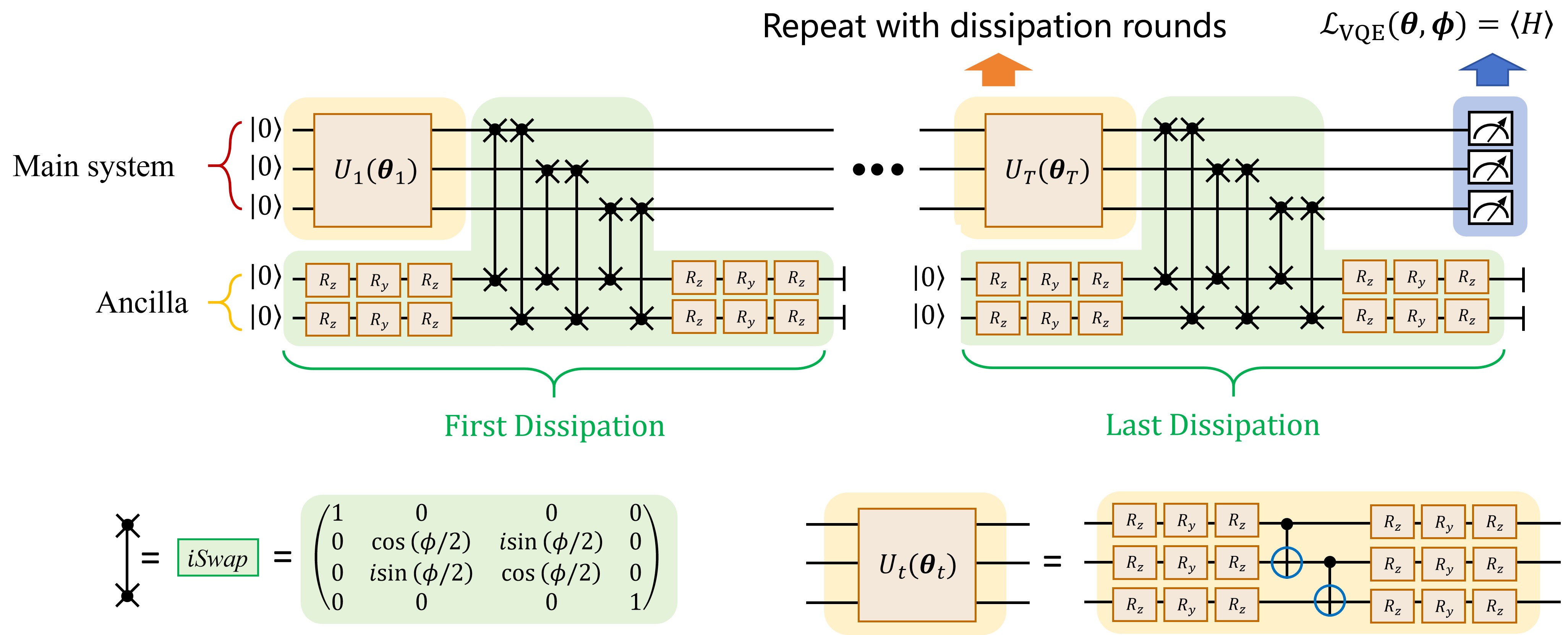}
    \caption{Architecture of the dissipative VQE model. The yellow blocks denote the system-only variational circuit, while the green blocks denote the dissipative modules implemented by system-ancilla couplings followed by ancilla trace-out and reset.}
    \label{NS1_DVQE}
\end{figure}

\begin{figure}[H]
    \centering
    \includegraphics[width=.9\linewidth]{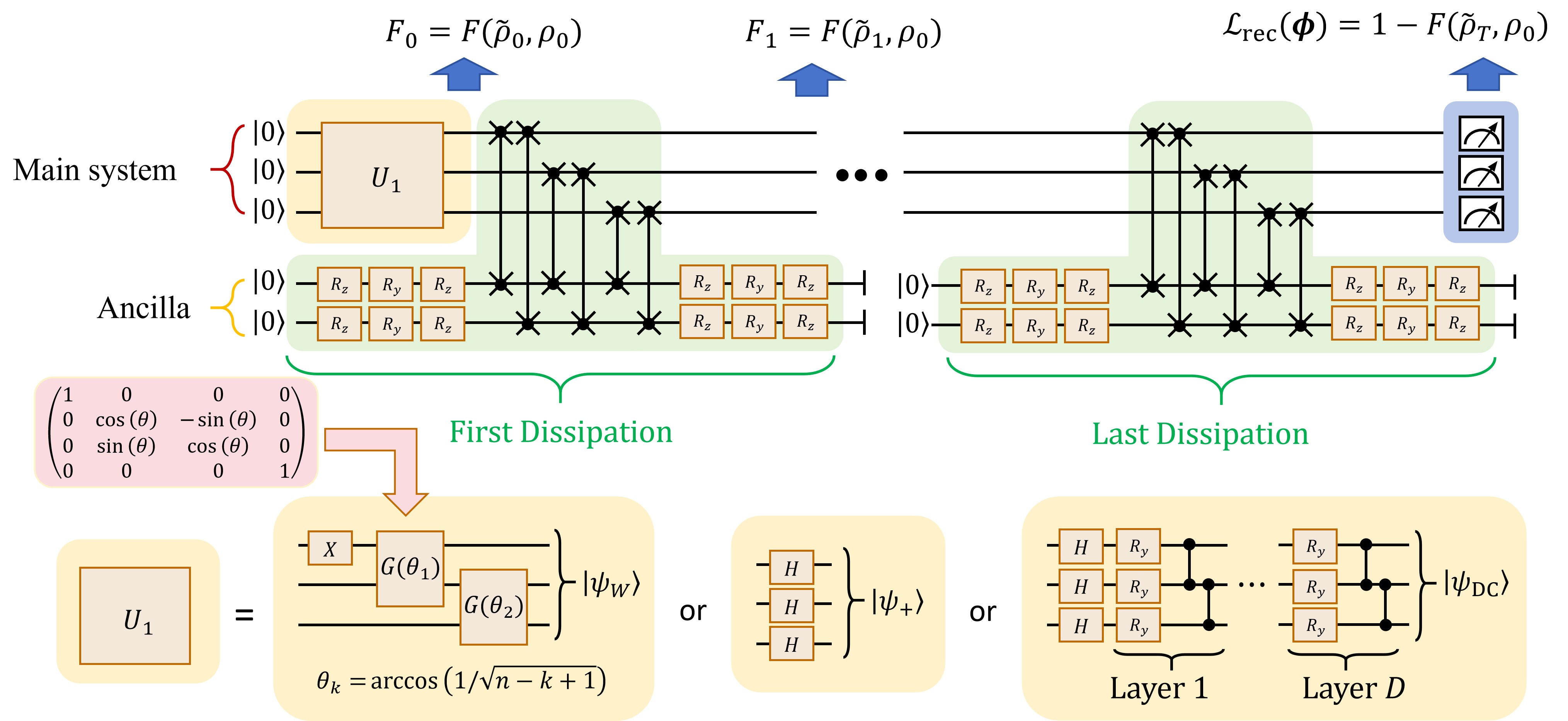}
    \caption{Architecture of the dissipative recovery model. A low-fidelity input state is fed into repeated dissipative rounds so that the reduced state of the main system approaches the target state.}
    \label{NS7_FidelitySchematic}
\end{figure}

In our implementation, the dissipative block is built from parameterized single-qubit rotations together with system-ancilla $i$SWAP couplings. Specifically, each round applies trainable local rotations to the involved qubits and then couples every system qubit to every ancilla qubit through an $i$SWAP gate. This parameterization allows the dissipative block to realize a flexible trainable system-ancilla interaction, whose reduced action on the main system defines the effective recovery channel.
Fig.~\ref{NS7_FidelitySchematic} shows the circuit-level realization of the dissipative recovery model used in our simulations. In contrast to the dissipative VQE architecture, the system-only variational layers are removed, and the input state is fed directly into repeated dissipative rounds. In each round, fresh ancillas initialized in $\ket{0}^{\otimes m}$ are coupled to the main system through the trainable dissipative block, after which the ancillas are traced out and reset. The repeated application of this ancilla-assisted channel allows the reduced state of the main system to be progressively reshaped toward the target state.

\subsection{Performance under different noise level}
Here, we provide the related extened numerical experiments of comparison of VQE and noise cancellation performance under different noise level. Fig.~\ref{NS3_VQENoiseStrength} complements the noise-channel comparison in the main text by showing the dependence on the noise strength in the VQE task. Panel (a) shows the standard VQE optimization, while panel (b) shows the dissipative VQE with five ancillas. Under the same noise conditions, the dissipative framework stays closer to the exact ground-state energy, indicating improved robustness against increasing noise strength. Fig.~\ref{NS12_FidelityThreshold} gives a complementary view of ancilla saturation through the fidelity gain in the $W$-state recovery task. The results show that the improvement is already maximized at about $m=2$, and that increasing the number of ancillas further does not produce a noticeable additional gain. In this sense, the recovery performance is not improved monotonically by enlarging the ancilla register, but instead exhibits a clear saturation behavior.

\begin{figure}[H]
    \centering
    \includegraphics[width=.9\linewidth]{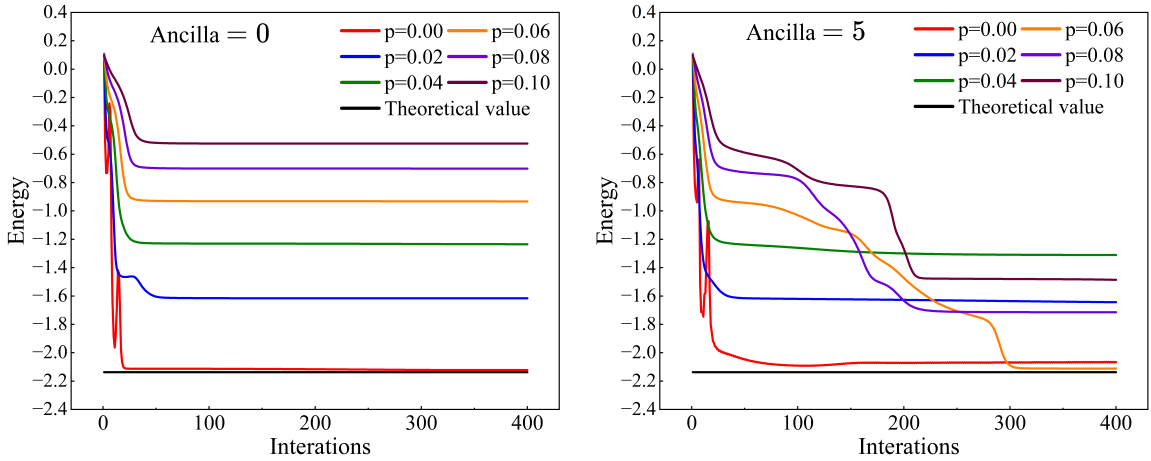}
    \caption{Comparison of VQE performance under different noise strengths. Panel (a) shows the standard ancilla-free VQE, for which the gap between the converged energy and the exact ground-state value grows with noise strength. Panel (b) shows the dissipative VQE, where the convergence remains faster and the final energies stay closer to the exact value, indicating enhanced robustness against noise.}
    \label{NS3_VQENoiseStrength}
\end{figure}

\begin{figure}[H]
    \centering
    \includegraphics[width=\linewidth]{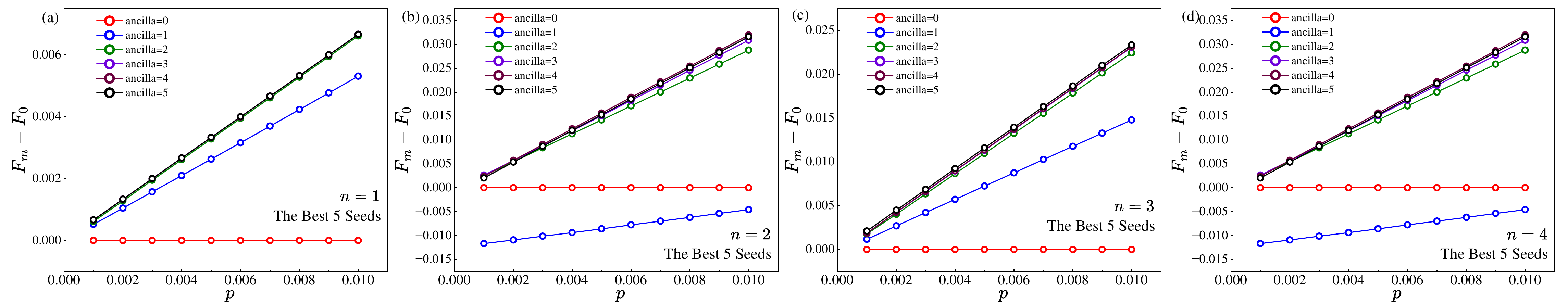}
    \caption{Fidelity gain $F_m-F_0$ as a function of the noise strength for $W$-state recovery, illustrating the effective ancilla saturation threshold $m=2$.}
    \label{NS12_FidelityThreshold}
\end{figure}

\section{Performance analysis of original VQE and dissipative VQE}
\label{sec:analysis}
\subsection{Loss functions for the original and dissipative VQAs}
Here, we consider the original VQA which prepares the ground state by using a parameterized unitary circuit, whereas the dissipative VQA through a parameterized noisy channel with periodic non-unital reset operations, as in Ref~\cite{zapusek2025scaling}.

\begin{definition}[Original VQA]
Given a Hamitlonian H, let $\rho_0$ be the input state and let $U(\theta_{\uni})$ be a parameterized unitary ansatz constructed via a product of parameterized gates of the form $U(\theta_{\uni})=\prod_j \exp(-i\theta_{\uni} H_j)$. The parametrized state of the original VQA is
\begin{equation}
  \rho_{\uni}(\theta_{\uni})
  :=
  U(\theta_{\uni})\rho_0 U(\theta_{\uni})^\dagger .
  \label{eq:unitary-state}
\end{equation}
The corresponding cost function is
\begin{equation}
  C_{\uni}(\theta_{\uni})
  :=
  \Tr\!\left[O\rho_{\uni}(\theta_{\uni})\right]
  =
\Tr\!\left[O\,U(\theta_{\uni})\rho_0U(\theta_{\uni})^\dagger\right].
  \label{eq:unitary-cost}
\end{equation}
\end{definition}
If unital noise is included, the preparation map may be replaced by a noisy unital channel applied to the unitary circuit, this modification can lead to noise-induced barren plateaus~\cite{wang2021noise}.

\begin{definition}[Dissipative VQA]
Given a Hamitlonian $H$ and input state $\rho_0$, the dissipative VQA utilizes reset-assisted non-unital quantum channel~\cite{zapusek2025scaling}. Concretely, let $\mathcal U^k_{\theta_{\dis}}$ denote the parameterized two-qubit layer at layer $k$, $\mathcal V^k$ a layer of fixed single-qubit gates, and $\mathcal N$ a tensor product of contractive unital single-qubit noise channels. The protocol contains reset channel $A_{n_r}$ expressed as $\mathcal A_{n_r}(\rho)=(1-q)\rho+q\,\Bigl(\Tr_R[\rho]\otimes |0\rangle\!\langle 0|_R\Bigr)$, which are interleaved with $L$ layers of noisy parameterized evolution $\{\mathcal V^k\circ \mathcal N\circ \mathcal U^k_{\theta_{\dis}}\}_{k=1}^L$.
With $\mathcal A_{n_r}^0=\Id$, the dissipative channel is
\begin{equation}
  \Phi_{\theta_{\dis}}^M
  :={\prod_{k=1}^{L}}
  \left(
    \mathcal V^k\circ \mathcal N\circ \mathcal U^k_{\theta_{\dis}}
    \circ \mathcal A_{n_r}^{\chi(k)}
  \right),
  \label{eq:dissipative-channel}
\end{equation}
where 
\begin{equation}
  \chi(k)=
  \begin{cases}
    1, & k\equiv 0 \pmod L,\\
    0, & \text{otherwise}.
  \end{cases}
\end{equation}
The variational state generated by dissipative VQA can be written as
\begin{equation}
  \rho_{\dis}(\theta_{\dis})
  :=
  \Phi_{\theta_{\dis}}^M(\rho_0),
  \label{eq:dissipative-state}
\end{equation}
and the corresponding cost function is
\begin{equation}
  C_{\dis}(\theta_{\dis})
  :=
  \Tr\!\left[O\rho_{\dis}(\theta_{\dis})\right]
  =
  \Tr\!\left[O\,\Phi_{\theta_{\dis}}^M(\rho_0)\right].
  \label{eq:dissipative-cost}
\end{equation}
\end{definition}

The reset channel $\mathcal A_{n_r}$ resets the qubits to $|0\rangle$ with probability $q\in(0,1]$, and acts as the identity with probability $1-q$. Here, we mainly consider the case of $q=1$. The key distinction from \eqref{eq:unitary-cost} is the non-unital reset-assisted map $\Phi_{\theta_{\dis}}^M$, which can remove entropy during the computation and is the mechanism behind the dissipative barren-plateau avoidance result of Ref.~\cite{zapusek2025scaling}.

\subsection{Optimization process and related assumptions}
Consider the optimization dynamics at step $t$ for both the original VQA and dissipative VQA under gradient descent, which can be written in the following form:
\begin{equation}
  \theta_{A,t+1}
  =
  \theta_{A,t}-\eta_A\nabla C_A(\theta_{A,t}).
\end{equation}
where $A\in\{\uni,\dis\}$ refers to original VQA and dissipative VQA , respectively.
Define the one-step decrease and the optimality gap by
\begin{equation}
  \Delta C_{A,t}
  :=
  C_A(\theta_{A,t})-C_A(\theta_{A,t+1}),
  \qquad
  \gap_{A,t}:=C_A(\theta_{A,t})-C_A^*,
  \qquad
  C_A^*:=\inf_\theta C_A(\theta).
\end{equation}
The global comparison concerns how fast each method reduces the corresponding gap $\gap_{A,t}$.

\begin{assumption}[Smoothness]\label{asm:smooth}
For $A\in\{\uni,\dis\}$ and $\forall \theta,\phi\in \mathbb{R}^p$, the cost $C_A$ is $\beta_A$-smooth, that is
\begin{equation}
  \norm{\nabla C_A(\theta)-\nabla C_A(\phi)}_2
  \le
  \beta_A\norm{\theta-\phi}_2.
\end{equation}
For gradient descent, it uses $0<\eta_A\le 1/\beta_A$.
\end{assumption}
Equivalently $\norm{\nabla^2 C_A(\theta)}_2\leq\beta_A$ for all $\theta$.
Gradient descent uses step size $0<\eta_A\leq 1/\beta_A$.
We additionally assume $\beta_A = O(\poly(n))$, which holds for local
Hamiltonians with $O$-norm $\norm{O}_\infty=O(\poly(n))$.

\begin{assumption}[Two-design]
\label{asm:two_design}
At initialisation $\theta_{A,0}$, the parameterised unitaries are drawn
from a distribution forming a unitary $2$-design on each gate.
In particular, $\E[\partial_\mu C_A(\theta_{A,0})]=0$ for all $\mu$,
so that $\E[(\partial_\mu C_A)^2]=\Var[\partial_\mu C_A]$ at $t=0$.
\end{assumption}

\begin{assumption}[PL condition for the dissipative VQA ]
\label{asm:PL}
There exists $\mu_{\mathrm{dis}}>0$ such that, for all $\theta$,
\begin{equation}
    \tfrac{1}{2}\norm{\nabla C_{\mathrm{dis}}(\theta)}_2^2
    \;\geq\;
    \mu_{\mathrm{dis}}\bigl(C_{\mathrm{dis}}(\theta)-C_{\mathrm{dis}}^*\bigr).
    \label{eq:PL}
\end{equation}
\end{assumption}

\begin{assumption}[Trajectory-uniform flatness of the original VQA]
\label{asm:flat}
There exist constants $B_{\mathrm{uni}}>0$ and $b>1$ such that, for all
$t\geq 0$ until the expected gap reaches a target $\varepsilon$,
\begin{equation}
    \E\bigl[\norm{\nabla C_{\mathrm{uni}}(\theta_{\mathrm{uni},t})}_2^2\bigr]
    \;\leq\; G_{\mathrm{uni}}(n)
    \;:=\; m\,B_{\mathrm{uni}}\,b^{-n},
    \label{eq:flat}
\end{equation}
where $m$ is the number of variational parameters.
\end{assumption}

\subsection{Conditional global convergence and separation}

In this section, we first establish the descent tools needed for the global convergence analysis. Lemmas~\ref{lem:lower-descent} and~\ref{lem:pl-contraction} are used to prove Theorem~\ref{thm:diss-global}, which gives a conditional global convergence guarantee for the dissipative VQA under a PL-type assumption. We then use Lemmas~\ref{lem:upper-descent} and~\ref{lem:unitary-step-lower} to derive a lower bound on the number of gradient-descent steps required by a trajectory-flat original VQA. Finally, Theorem~\ref{thm:global-separation} combines these two results to establish a conditional separation in iteration complexity between the dissipative and original VQAs.

\begin{lemma}[Lower bound of per-step descent]
\label{lem:lower-descent}
Let \(f:\mathbb R^p\to\mathbb R\) be \(\beta\)-smooth and let \(x^+=x-\eta\nabla f(x)\) with \(0<\eta\le1/\beta\).  Then
\begin{equation}
  f(x)-f(x^+)
  \ge
  \frac{\eta}{2}\norm{\nabla f(x)}_2^2 .
  \label{eq:lower-descent}
\end{equation}
\end{lemma}

\begin{proof}
By \(\beta\)-smoothness,
\begin{equation}
  f(x^+)
  \le
  f(x)-\eta\norm{\nabla f(x)}_2^2
  +\frac{\beta\eta^2}{2}\norm{\nabla f(x)}_2^2.
\end{equation}
Since \(\eta\le1/\beta\), we have \(\eta(1-\beta\eta/2)\ge\eta/2\).  Rearranging gives the claim.
\end{proof}

\begin{lemma}[Upper bound of per-step descent]
\label{lem:upper-descent}
Assume \(f\) is twice differentiable and \(\norm{\nabla^2 f(z)}_2\le\beta\) on the segment between \(x\) and \(x^+=x-\eta\nabla f(x)\).  If \(0<\eta\le1/\beta\), then
\begin{equation}
  f(x)-f(x^+)
  \le
  \frac{3\eta}{2}\norm{\nabla f(x)}_2^2 .
  \label{eq:upper-descent}
\end{equation}
\end{lemma}

\begin{proof}
The lower Taylor estimate implied by \(\norm{\nabla^2 f(z)}_2\le\beta\) gives
\begin{equation}
  f(x^+)
  \ge
  f(x)-\eta\norm{\nabla f(x)}_2^2
  -\frac{\beta\eta^2}{2}\norm{\nabla f(x)}_2^2.
\end{equation}
Thus
\begin{equation}
  f(x)-f(x^+)
  \le
  \eta\left(1+\frac{\beta\eta}{2}\right)\norm{\nabla f(x)}_2^2
  \le
  \frac{3\eta}{2}\norm{\nabla f(x)}_2^2,
\end{equation}
because \(\eta\le1/\beta\).
\end{proof}

\begin{lemma}[PL contraction]
\label{lem:pl-contraction}
Suppose \(f\) is \(\beta\)-smooth and satisfies the PL inequality
\begin{equation}
  \frac12\norm{\nabla f(x)}_2^2
  \ge
  \mu\bigl(f(x)-f^*\bigr)
  \qquad\text{for all }x.
\end{equation}
Then gradient descent with \(0<\eta\le1/\beta\) satisfies
\begin{equation}
  f(x_t)-f^*
  \le
  e^{-\eta\mu t}\bigl(f(x_0)-f^*\bigr).
  \label{eq:pl-rate}
\end{equation}
\end{lemma}

\begin{proof}
Lemma~\ref{lem:lower-descent} gives
\begin{equation}
  f(x_{t+1})-f^*
  \le
  f(x_t)-f^*-\frac{\eta}{2}\norm{\nabla f(x_t)}_2^2.
\end{equation}
Using the PL inequality,
\begin{equation}
  f(x_{t+1})-f^*
  \le
  (1-\eta\mu)\bigl(f(x_t)-f^*\bigr)
  \le
  e^{-\eta\mu}\bigl(f(x_t)-f^*\bigr).
\end{equation}
Iterating proves \eqref{eq:pl-rate}.
\end{proof}

\begin{lemma}[Trajectory-flat original VQA requires many steps]
\label{lem:unitary-step-lower}
Assume Assumptions~\ref{asm:smooth} and \ref{asm:flat}.  Suppose \(\E[\gap_{\uni,0}]\ge\delta_0>0\).  If \(\E[\gap_{\uni,T}]\le\varepsilon<\delta_0\), then
\begin{equation}
  T
  \ge
  \frac{2(\delta_0-\varepsilon)}{3\eta_{\uni}G_{\uni}(n)}
  =
  \Omega\!\left(\frac{b^n}{\eta_{\uni}m}\right),
  \label{eq:unitary-step-lower}
\end{equation}
\end{lemma}

\begin{proof}
By Lemma~\ref{lem:upper-descent} and Assumption~\ref{asm:flat}, 
\begin{equation}
  \E[\Delta C_{\uni,t}]
  \le
  \frac{3\eta_{\uni}}{2}
  \E\norm{\nabla C_{\uni}(\theta_{\uni,t})}_2^2
  \le
  \frac{3\eta_{\uni}}{2}G_{\uni}(n).
\end{equation}
Summing from \(t=0\) to \(T-1\) gives
\begin{equation}
  \E\bigl[C_{\uni}(\theta_{\uni,0})-C_{\uni}(\theta_{\uni,T})\bigr]
  \le
  \frac{3\eta_{\uni}}{2}T G_{\uni}(n).
\end{equation}
The left-hand side equals \(\E[\gap_{\uni,0}]-\E[\gap_{\uni,T}]\ge\delta_0-\varepsilon\). By rearranging it, then we have \eqref{eq:unitary-step-lower}.
\end{proof}

\begin{theorem}[Conditional global convergence of the dissipative VQA]
\label{thm:diss-global}
Under Assumptions~\ref{asm:smooth} and \ref{asm:PL}, gradient descent on
\begin{equation}
  C_{\dis}(\theta)=\Tr\!\left[O\Phi^M_\theta(\rho_0)\right]
\end{equation}
with step size \(\eta_{\dis}=1/\beta_{\dis}\) satisfies
\begin{equation}
  \gap_{\dis,t}
  \le
  e^{-\mu_{\dis}t/\beta_{\dis}}\gap_{\dis,0}.
  \label{eq:diss-global-rate}
\end{equation}
Consequently, the number of iterations required to reach \(\gap_{\dis,T}\le\varepsilon\) obeys
\begin{equation}
  T_{\dis}(\varepsilon)
  \le
  \frac{\beta_{\dis}}{\mu_{\dis}}
  \log\!\left(\frac{\gap_{\dis,0}}{\varepsilon}\right).
  \label{eq:diss-global-iterations}
\end{equation}
If \(\beta_{\dis}/\mu_{\dis}\le\poly(n)\) and \(\log(\gap_{\dis,0}/\varepsilon)\le\poly(n)\), then
\begin{equation}
  T_{\dis}(\varepsilon)
  \le
  \poly(n).
\end{equation}
\end{theorem}

\begin{proof}
Apply PL contraction, i.e. the following Lemma~\ref{lem:pl-contraction}, with \(f=C_{\dis}\), \(\eta=1/\beta_{\dis}\), and \(\mu=\mu_{\dis}\), and then by solving
\begin{equation}
  e^{-\mu_{\dis}T/\beta_{\dis}}\gap_{\dis,0}\le\varepsilon
\end{equation}
we have \eqref{eq:diss-global-iterations}.
\end{proof}

\begin{theorem}[Conditional global iteration-complexity separation]
\label{thm:global-separation}
Suppose the dissipative VQA satisfies the assumptions of Theorem~\ref{thm:diss-global} with
\begin{equation}
  \frac{\beta_{\dis}}{\mu_{\dis}}
  \le
  \poly(n).
\end{equation}
and also suppose that the original VQA satisfies Assumption~\ref{asm:flat}, with
\begin{equation}
  \eta_{\uni}m\le\poly(n),
  \qquad
  \E[\gap_{\uni,0}]\ge\delta_0>0.
\end{equation}
Then
\begin{equation}
  T_{\dis}(\varepsilon)
  \le
  \poly(n)\log\!\left(\frac{\gap_{\dis,0}}{\varepsilon}\right),
  \label{eq:global-sep-dis}
\end{equation}
whereas reducing the expected gap of the original VQA from \(\delta_0\) to any fixed \(\varepsilon<\delta_0\) requires
\begin{equation}
  T_{\uni}(\varepsilon)
  \ge
  \Omega\!\left(\frac{b^n}{\poly(n)}\right).
  \label{eq:global-sep-uni}
\end{equation}
Thus, under the additional PL and trajectory-uniform flatness hypotheses, the dissipative VQA has a conditional exponential advantage in gradient-descent iteration complexity.
\end{theorem}

\begin{proof}
The dissipative upper bound \eqref{eq:global-sep-dis} follows from Theorem~\ref{thm:diss-global}, and with the assumption \(\eta_{\uni}m\le\poly(n)\), we have the unitary lower bound \eqref{eq:global-sep-uni} follows from the Lemma~\ref{lem:unitary-step-lower}.

\end{proof}

The dissipative VQA does not automatically have proven global convergence from the barren-plateau-avoidance paper alone. However, if the dissipative objective additionally satisfies a global PL condition with polynomial constants, and if the original VQA remains trajectory-uniformly flat, then dissipative VQA has polynomial iteration complexity while the original VQA needs exponentially many steps.

\end{document}